\documentclass[aip,jcp,a4paper,12pt,nofootinbib]{revtex4-1}

\usepackage[utf8]{inputenc}
\usepackage[T1]{fontenc}
\usepackage{mathtools}
\usepackage{amsmath,amsthm}
\usepackage{amssymb}
\usepackage{physics}
\usepackage{subcaption}

\newtheorem{thm}{Theorem}
\newtheorem{rem}{Remark}
\newtheorem{lemma}{Lemma}

\newcommand{\np}[1]{#1}
\newcommand\tiret{\nobreakdash-\hspace{0pt}}

\newcommand{\vs}[1]{{\mathsf{#1}}}
\newcommand{\vc}[1]{\boldsymbol{\mathbf{#1}}}
\newcommand{\fld}{\vc{E}}
\newcommand{\av}{\vc{a}}
\newcommand{\mv}{\vc{m}}
\newcommand{\nv}{\vc{n}}
\newcommand{\rv}{\vc{r}}
\newcommand{\rijn}{\rv_{ij\nv}}

\newcommand{\vv}{\vc{v}}
\newcommand{\Dif}{\vs{D}}

\newcommand{\ii}{{\rm i}}

\newcommand{\Li}{\vs L_i}
\newcommand{\Lj}{\vs L_j}
\newcommand{\oS}{\mathcal S}

\newcommand{\urvN}{\underbar{r}^{\{N\}}}
\newcommand{\param}{\alpha^{1/2}}

\newcommand{\N}{\mathbb{N}}
\newcommand{\Z}{\mathbb{Z}}
\newcommand{\R}{\mathbb{R}}

\DeclareMathOperator{\erfc}{erfc}
\newcommand{\Fc}{\mathcal{F}}
\newcommand\ordre{\mathcal{O}}

\newcommand{\mEwald}{\mathsf{Ewald}}
\newcommand{\mself}{\mathsf{self}}
\newcommand{\msurf}{\mathsf{surf}}

\newcommand{\Ewald}{Ewald}

\newcommand{\pme}{\textsc{pme}}

\DeclarePairedDelimiterX\Set[2]{\lbrace}{\rbrace}{ #1 \,\delimsize|\, #2 }

\usepackage[usenames,dvipsnames,svgnames]{xcolor}

\newcommand{\old}[1]{}

\newcommand{\black}{\color{black}}

\newcommand{\refO}[1]{{#1}}
\newcommand{\refT}[1]{{#1}}

\usepackage{filecontents}

\begin{document}
\title{A coherent derivation of the \Ewald\ summation for arbitrary orders of multipoles:
The self-terms}
\author{Benjamin Stamm}
\affiliation{Center for Computational Engineering Science, RWTH Aachen University, Aachen,
  Germany}
\author{Louis Lagardère}
\affiliation{Sorbonne Universit\'e, Institut des Sciences du Calcul et des Donn\'ees, Paris, France}
\affiliation{Sorbonne Universit\'e, Institut Parisien de Chimie Physique et Th\'eorique, FR 2622 CNRS, Paris, France}
\affiliation{Sorbonne Universit\'e, Laboratoire de Chimie Th\'eorique, UMR 7616 CNRS, Paris, France}
\author{Étienne Polack}
\affiliation{Sorbonne Universit\'e, Laboratoire de Chimie Th\'eorique, UMR 7616 CNRS, Paris, France}
\affiliation{Sorbonne Université, Université Paris-Diderot SPC, CNRS, Laboratoire Jacques-Louis Lions, LJLL, F-75005 Paris}
\author{Yvon Maday}
\affiliation{Sorbonne Université, Université Paris-Diderot SPC, CNRS, Laboratoire Jacques-Louis Lions, LJLL, F-75005 Paris}
\affiliation{Institut Universitaire de France, Paris, France}
\affiliation{Brown Univ, Division of Applied Maths, Providence, RI, USA}
\author{Jean-Philip Piquemal}
\affiliation{Sorbonne Universit\'e, Laboratoire de Chimie Th\'eorique, UMR 7616 CNRS, Paris, France}
\affiliation{Institut Universitaire de France, Paris, France}
\affiliation{The University of Texas at Austin, Department of Biomedical Engineering, TX, USA}
\begin{abstract}
  In this work, we provide the mathematical elements we think essential for a proper
  understanding of the calculus of the electrostatic energy of point-multipoles
  of arbitrary order under periodic boundary conditions.
  The emphasis is put on the expressions of the so-called self parts of the \Ewald\,
  summation where different expressions can be found in literature.
 Indeed, such expressions are of prime importance in the context of new generation polarizable force field where the self field appears in the polarization equations.
  We provide a general framework, where the idea of the \Ewald\ splitting is applied to the
  electric potential and subsequently, all other quantities such as the electric field, the
  energy and the forces are derived consistently thereof.
  Mathematical well-posedness is shown for all these contributions for any order of multipolar distribution.

\end{abstract}

\maketitle

\section*{Introduction}


The computation of physical quantities involving the \np{Coulomb} potential is a challenging
issue due to the slow decay of the interacting kernel as the inverse of the distance.
This long-range potential often prevents the use of simple techniques like cutoffs methods
that only take into account short-range interactions.
This problem has been addressed with the use of hierarchical methods (of order
\( \ordre(N) \) or \( \ordre(N\log N) \) complexity) that approximate the long-range
interactions and \np{Fourier} (of order \( \ordre(N\log N) \)) methods that compute part of
the \np{Coulomb} interaction in the dual space by considering the physical
system under periodic boundary conditions.

For molecular dynamics simulations of biological systems, the most widely used method is a
\np{Fourier} method, the particle-mesh \Ewald\cite{darden93, essmann95} --- or shortly \pme.
This method is based on the \Ewald\ summation\cite{ewald21}, which gives a well-posed
definition for the energy of the system.
This is indeed not granted at all, since the energy is not well defined due to the
conditional convergence of the involved series of the infinite periodic system if the
(neutral) unit cell has a non-zero dipolar moment.
In this case, different orders of summation provide different energies.

\subparagraph{Background on the \Ewald\ summation.}

The mathematical derivation of the \Ewald\ energy summation for point charges in three
dimensions was carried out by~\citet{redlack75, deleeuw80}.
With respect to the focus of this paper involving multipoles of any order,
\citet{ weenk1977calculation} and \citet{smith98} gave expressions for the energy using \Ewald\ summation for
density of charges expressed as a sum of multipoles up to quadrupoles.
Those expressions have been used, for example, in the works of of~\citet{nymand00,
toukmaji00, wang05} for dipoles and by~\citet{aguado03} for quadrupoles.

However some expressions in the paper by~\citet{smith98} are justified using physical
insight, and only the \Ewald\ energies and forces are given.
We think this is the reason why some other authors use other (inconsistent) expressions.
For example~\citet{nymand00} give an expression for the electric field that is different
from the one by~\citet{toukmaji00}.
This difference was then discussed by \citet{laino2008notes} and corrected in \citet{stenhammar2011some}.

Moreover all the terms (potential, field, energy, forces) for the \Ewald\ summation are in
our knowledge never presented all together in one place consistently, and the derivation is
seldom explained.
For example, \citet{aguado03, wang05} don't give expression for the field
and~\citet{toukmaji00} give no expression for the potential.
\citet{stenhammar2011some} builds an exception, however, the proposed self-energy differs for \refO{quadrupolar} distributions as the work by \citet{stenhammar2011some} does not include a quadruple-\refO{quadrupol} interaction whereas \citet{aguado03} \refO{does. The latter is however with a different formula than what we propose later in this work.}
This may be explained by a missing double factorial in \citet{aguado03} and \citet{nymand00} that was pointed out by \citet{laino2008notes}.
As only the net expressions are provided, it is difficult to trace back this difference.
\refT{
Recent developments have been made for efficient PME calculations using spherical harmonic point multipoles in \citet{Giese2015} and \citet{Simmonett2014}, where in particular the former also provides expressions for  energies, potentials, and forces using arbitrary order point multipoles.
}

\subparagraph{Contribution.}

This paper should be seen as an extension of the work of \citet{smith98}.
Although not fully rigorous and lead by physical intuition, his reasoning for the expression
of the self-energy can be proven with the use of some mathematical arguments, which can then
be used to find the self-terms of any multipolar distribution.
While we do not introduce a new theory, model or mathematical expressions, we introduce here
a coherent mathematical framework to derive the self-terms of multipolar distributions of
any order for the electric potential and field as well as the associated energy and forces and confirm the results proposed by \citet{smith98}.
Further, we present proofs of the  well-posedness of the self parts to the energy, electric potential and field for multipolar distributions of any order.

Our derivation is different from what has been proposed in the past, and emphasizes that the
\Ewald\ splitting should first be done on the potential or the field --- and not directly on
the energy.
We derive the self-potential and self-field from scratch using \Ewald\ splitting and deduce
from those expressions the results for the self-energy and self-forces.

The purpose of the present article is to provide a coherent mathematically driven
derivation of all self-terms, which, in consequence, provides a base for methodology
developments of force-fields.
We present in the appendix of a complete and precise derivation of all self-terms such that differences in expressions as highlighted above can be traced back.
In particular, and this is our main motivation, a correct derivation of the self-field is
indispensable for polarizable force-fields.
Indeed, to solve the polarization equation, the total field, and thus the self-field, is
required to compute the polarization field\cite{Lagardere2015ScalableEwald}.
In practice, such terms are well implemented in production codes like Tinker and Amber. However, other codes exist and omitting these terms would result in highly different properties. Indeed, as shown in Figure~\ref{fig:oxygen}, omitting the self-field in the computation of the polarization energy results in highly different oxygen-oxygen radial distribution function. 
Therefore, it is of prime importance for developers to have a robust justification of the expression to implement.
This is in contrast to non-polarizable force-fields where only the energy and forces are
needed to derive a correct dynamics.

\begin{figure}[htbp!]
  \centering
    \includegraphics[width=0.48\textwidth]{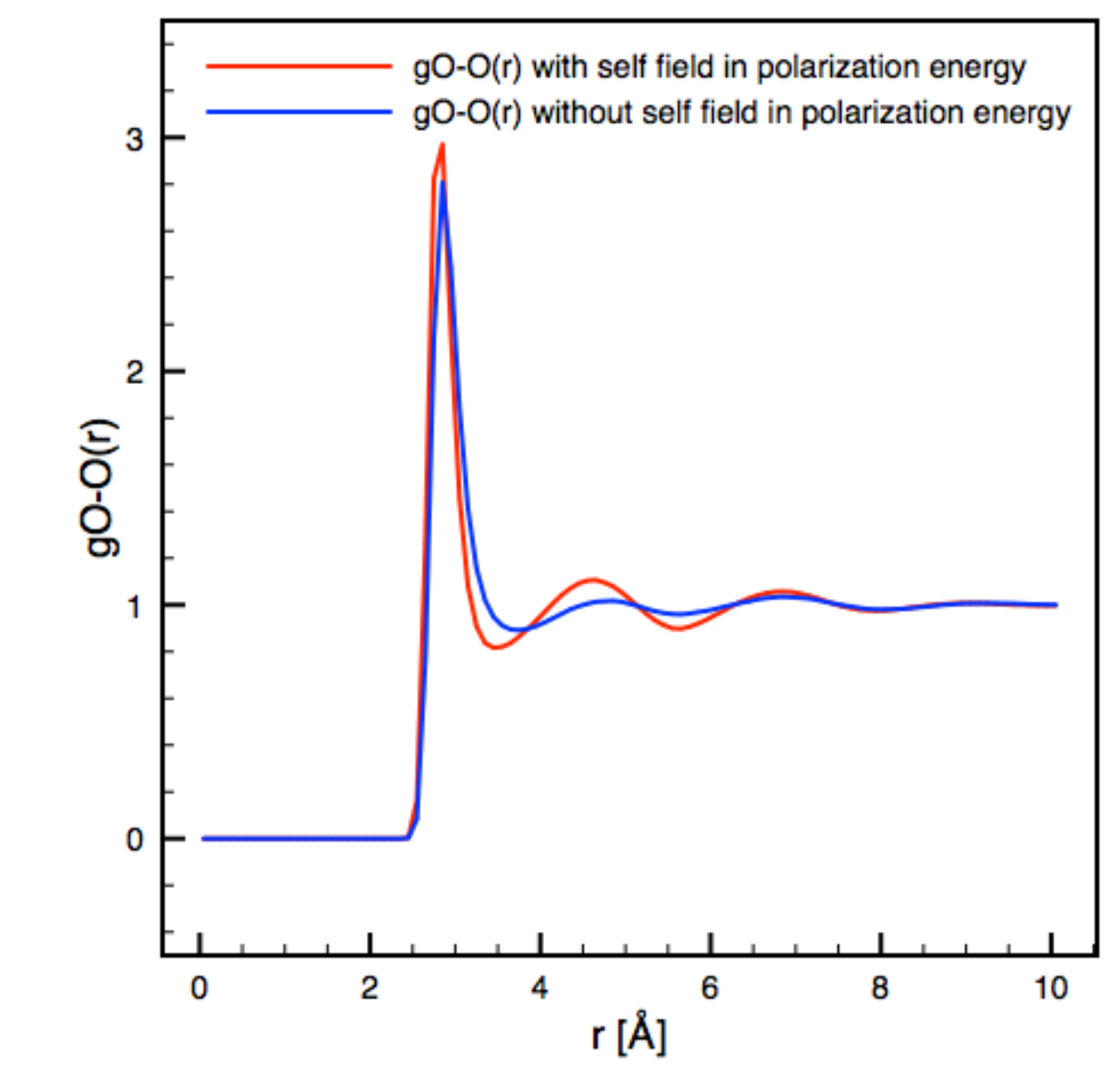}
  \caption{
  Computational experiment performed with the Tinker-HP \cite{lagardere2018tinker} software and the AMOEBA force field. Removing the self-field terms in the computation of the polarization energy gives rise to strong differences in the oxygen-oxygen radial distribution function compared to the correct Tinker-HP initial implementation. Simulation settings: 1ns NVT simulation at 300K, 4000 water molecules within a $49.323^3$ Angstrom square box.
  }\label{fig:oxygen}
\end{figure}

\subparagraph{Outline.}

First, in Section~\ref{sec:ewald-mult} we introduce the notations that we use and review
general results about the \Ewald\ summation.
In Section~\ref{sec:derivation-pot}, Section~\ref{sec:derivation-field} and
Section~\ref{sec:derivation-energy} we give, respectively, a derivation for the potential,
the field and the energy using \Ewald\ summation.
Finally, in Section~\ref{sec:self} we give explicit expressions of the self-terms and provide the proof that justifies the
existence of the self-terms for any multipolar distribution.

\section{\Ewald\ summation for multipoles}\label{sec:ewald-mult}

In this article, we consider a system composed of a discrete distribution of \( N \) point
multipoles  in \( \R^3 \) under periodic boundary conditions.
The system consists in an electronically neutral primitive triclinic cell \( U \) with
charges in form of multipoles located at \( \rv_{i} \in U\) for \( i \in \{1, \ldots, N\} \).
The set of positions \( \rv_{i} \) is represented by the global vector
\( \urvN \coloneqq (\rv_{1},\ldots, \rv_{N})\).
The unit cell \( U \) is then duplicated in all directions and the system derived from \( \urvN \) is
therefore composed of an infinity of charges.

The unit cell \( U \) is spanned by the three vectors \( (\av_1, \av_2, \av_3) \) which is
called the basis of \( U \).
We then introduce the lattice-indices \( \nv \) and \( \mv \) of the form
\begin{equation}
  \label{eq:nm_index}
  \nv = \sum_{1\leq \gamma\leq 3} n_{\gamma} \, \av_{\gamma} \quad \text{and} \quad
  \mv = \sum_{1\leq\gamma\leq 3} m_{\gamma} \, \av_{\gamma}^{*},
\end{equation}
where \( n_{\gamma}, m_{\gamma} \in \N \) and \( (\av_1^*, \av_2^*, \av_3^*) \) is the dual basis of
\( (\av_1, \av_2, \av_3) \): that is
\( \av_\gamma^* \cdot \av_{\gamma^{\prime}} = \delta_{\gamma\gamma^{\prime}} \) (the
\np{Kronecker} symbol).
We will also denote by \( V \) the volume of the primitive cell \( U \) and by \( U^* \) the dual of the primitive cell.

Then, one can informally introduce ``the'' electrostatic interaction energy of \( \urvN \) up to \( 2^{p} \)\tiret{}poles,
\( p \in \N \)  as
\begin{equation}
  \label{eq:multipoles}
  \mathcal{E}(\urvN)
  \coloneqq \frac{1}{2} \sideset{}{^{\prime}}\sum_{\substack{\nv \\ 1\leq i,j\leq N}}
  \Li \, \Lj  \, \frac{1}{\abs{\rijn}},
\end{equation}
where \( \rijn \coloneqq \rv_{i} - \rv_{j} + \nv \), the sign \( {}^{\prime} \) on the sum means that for \( i = j \) when \( \nv = 0 \) the
interaction is not counted (this avoids self-interaction of a point multipole with itself) and the multipolar operator
\( \Li \) is defined as
\begin{equation}
  \Li \coloneqq \sum_{0\leq k \leq p} \vs M^{k}_{i} \cdot \Dif^{k}_{i}.
\end{equation}
Here, \( \vs M^{k}_{i} \) is a \( k \)\tiret{}dimensional array of dimension
\( 3^{k} \) describing moment of the point \( 2^{k} \)\tiret{}pole, \( \Dif^{k}_{i} \) is the
matrix of \( k \)\tiret{}order partial derivatives with respect to the variable \( \rv_i \)
and \( \cdot \) is the point-wise product which writes 
\[
	(\vs A \cdot \vs B)_\alpha 
	= \vs A_\alpha \, \vs B_\alpha 
	= \vs A_{\alpha_1,\ldots,\alpha_k} \, \vs B_{\alpha_1,\ldots,\alpha_k},
\]
for two arbitrary $k$-dimensional arrays $\vs A, \vs B \in \mathbb R^{3k}$ and where $\alpha=(\alpha_1,\ldots,\alpha_k)$, $\alpha_i\in \{1,2,3\}$, is a $k$-dimensional multi-index.

For instance, \( k=0 \) represents a point-charge of charge \( \vs M^{0}_{i} \) at \( \rv_i \) and \( k=1 \) a dipole where \( \Dif^{1}_{i} \) is equivalent to the usual \( \nabla \) notation
with respect to \( \rv_i \) and \( \vs M^{1}_{i} \) denotes the dipolar moment for each
location  \( \rv_i \).
Next, \( k=2 \) represents a quadrupole, \( \Dif^{2}_{i} \) denotes the Hessian matrix and \( \vs M^{2}_{i} \) is
a \( 3\times 3 \) matrix that incorporates the quadrupolar moments.

We will see that the energy in \eqref{eq:multipoles} is actually {\it not well defined}: As in the case of single point charges, it can be shown, by a \np{Taylor} expansion with
  respect to \( \nv \), that the series in equation~\eqref{eq:multipoles} is what is called
  conditionally convergent.
  That implies that the result of the energy \( \mathcal{E}(\urvN) \) depends on the order of
  summation and is thus not uniquely defined.

The electrostatic energy can equivalently be stated in the following form
\begin{equation}\label{eq:en_fld}
  \mathcal{E}(\urvN) = \frac{1}{2} \sum_{1\leq i \leq N} \Li  \phi^{i}(\rv_{i})
  \qquad \text{with} \qquad
  \phi^{i}(\rv_{i}) \coloneqq \sideset{}{^{\prime}}\sum_{\substack{\nv\\1\leq j\leq N}}
  \Lj \frac{1}{\abs{\rijn}},
\end{equation}
so that \( \phi^{i}(\rv_{i}) \) denotes the potential at \( \rv_i\) which is
generated by all multipoles different than the one located at \( \rv_i \).
In consequence, equation~\eqref{eq:multipoles} represents indeed the interaction energy between
every multipole \( i \) in the unit cell with the potential created by all other multipoles
(indexed by \( j \) and \( \nv \)) of the infinite lattice.

Let us make a subtle comment. While $\rv_i$ is the fixed position of the $i$-th multipole, the multipole operator $\Li$ involves derivatives which requires to consider the potential $\phi^i$ in a local neighborhood of $\rv_i$. We denote therefore by $\rv$ the variable belonging to a local neighborhood of $\rv_i$ and write
\begin{equation}
	\label{eq:Equiv}
	\Li  \phi^i(\rv_i)
	 =
	\vs M^{k}_{i} \cdot \big(  \Dif^{k}_{\rv } \phi^i(\rv)\big)\big|_{\rv=\rv_i},
\end{equation}
since we have to consider the potential $\phi^i(\rv)$ and its derivatives ultimately evaluated at $\rv=\rv_i$.

As anticipated above, equations~\eqref{eq:multipoles} and \eqref{eq:en_fld} are not well-defined and hence the need to use a definition of an expression for the energy that is well defined.
One possible remedy is the introduction of the \Ewald\ energy to give a unique meaning of
this expression by
\begin{multline}
  \label{eq:ewald1}
  \mathcal{E}_{\mEwald}(\urvN)
  \coloneqq \frac{1}{2} \sideset{}{^{\prime}}\sum_{\substack{\nv\\1\leq i,j\leq N}}
  \Li \, \Lj \left( \frac{\erfc(\alpha^{1/2}\abs{\rijn})}{\abs{\rijn}} \right) \\
  + \frac{1}{2\pi V} \sum_{\mv \neq \vc 0} \frac{\exp(-\pi^2\mv^2/\alpha)}{\mv^2}\oS(\mv)\oS(-\mv)
  +  \mathcal{E}_{\mself}(\urvN),
\end{multline}
where \( \alpha \) is a positive real number,
\begin{equation}\label{eq:struct_fourier}
  \oS(\mv) \coloneqq \sum_{1\leq j \leq N} \Fc(\Lj)(\mv) \exp(2\pi \ii \mv \cdot \rv_{j})
\end{equation}
is the structure factor and \( \Fc \) is the discrete \np{Fourier} transform of the operator
\( \Lj \).
For example for a point-multipoles up to order $p=2$ (quadrupoles), $\Fc$ reads
\[
	\Fc(\Lj)(\mv)
	=
	\vs M^{0}_{j}
	+
	2\pi  \ii \vs M^{1}_{j} \cdot \mv
	-
	(2\pi)^2 \vs M^{2}_{j} \cdot M,
\]
with $M_{\gamma\gamma'} = m_\gamma m_{\gamma'}$.
The third term in~\eqref{eq:ewald1} is commonly referred to as the self-energy.
In the realm of polarizable force-fields, the commonly used definition of the self-energy is
the one from~\citet{smith98}.

The fundamental property of the \Ewald\ energy is that it is independent on the order of
summation due of the absolute convergence of the involved sums.

It can be shown\cite{deleeuw80, darden08} that the interaction energy~\eqref{eq:multipoles}
of the system is related to the \Ewald\ energy through the relation
\begin{equation}
  \label{eq:energy}
  \mathcal{E}(\urvN) = \mathcal{E}_{\mEwald} + J(\vs D, \vs M),
\end{equation}
where the surface term \( J(\vs D, \vs M) \) depends on the dipolar moment
\( \vs D = \sum_{1\leq i\leq N} \vs M^{0}_{i} \, \rv_{i} \) and the sum of dipoles
\( \vs M = \sum_{1\leq i\leq N} \vs M^{1}_{i} \) of the primitive cell \( U \).
Only this term is responsible for the order of summation in equation~\eqref{eq:multipoles}, it reflects the macroscopic shape of the system (see the upcoming Remark~\ref{rem:ms} for a discussion on the notion of macroscopic shape).
The order of summation of the conditionally convergent series is therefore a factor to
choose in order to specify the exact value of the interaction energy \( \mathcal{E}(\urvN) \) and 
is often supposed to be spherical (by shells of \( \nv \) such that
\( \abs{\nv} \) is increasing).

By supposing that the macroscopic system is surrounded by a continuum dielectric
with some dielectric permittivity \( \varepsilon \), the interaction of the microscopic
system with the continuum can be taken into account and explicitly dealt with for spherical
summation orders.
Further, in the limiting case of a perfect conductor \( \varepsilon=\infty \) as surrounding
environment (and still with spherical summation order), it can be proven that the
surface term vanishes\cite{de1980simulation,darden08}.
This model is called the tinfoil model.
In consequence, this implies that the energy of the system is in this case the \Ewald\
energy.

In this paper, we do not longer comment on the convergence issues, which will be subject of a forthcoming paper, and concentrate on the proper definition of the self-energy \( \mathcal{E}_{\mself}(\urvN) \), which requires some
subtle development if general multipoles are considered that go beyond the results for
point charges.

More precisely, there are two aspects that we address in the work.
First, we investigate a mathematically clean derivation of the self-potential (and thus of
the energy thanks to equation~\eqref{eq:en_fld}) and self-field when general multipoles are
considered and not only point-charges.
We then deduce thereof the expression of the self-energy.
Second, we present the proofs which demonstrates that these quantities are mathematically
well-defined.

\section{Derivation of the potential}\label{sec:derivation-pot}

First, we revisit the derivation of the \Ewald\ summation for the potential generated by the
multipoles.
The conditionally convergent series in \eqref{eq:en_fld} defining the potential $\phi^j$ is given a precise meaning by considering the limit
\[
	\lim_{k\to \infty}
    \sideset{}{^{\prime}}\sum_{\substack{\nv\in \Omega(P,k) \\ 1\leq i\leq N}}
  	\Li \frac{1}{\abs{\rv_{i} - \rv_{j} + \nv}},
\]
for some domain $P$ in \( \R^{3} \)  containing the origin that represents the macroscopic shape of the system (see Remark~\ref{rem:ms}) and where
\begin{equation}
  \Omega(P, k) \coloneqq
  \Set*{\nv = \sum_{1 \leq \gamma \leq 3} n_\gamma \av_\gamma}%
  {{(n_\gamma)}_{1 \leq \gamma \leq 3} \in \Z^3, \; \frac{\nv}{k} \in P}.
\end{equation}
At the base of the derivation of the potential is the splitting
\begin{equation}
  \label{eq:distanceSplit}
  \frac{1}{\abs{\rv}}
  = \frac{\erfc(\param \abs{\rv})}{\abs{\rv}}
  + \frac{1}{\pi} \sum_{\mv} \int_{U^*} \frac{\exp(-\pi^2 \abs{\vv + \mv }^2 / \alpha)}{\abs{\vv + \mv}^2}
  \exp(-2i\pi(\vv + \mv) \cdot \rv) \, \dd[3] \vv,
\end{equation}
for any positive \( \alpha \) and which can be deduced\cite{darden08} from the integral expression of the gamma function at the point \( \frac12 \) for all \( \rv \) but at the origin.

Using the present splitting and following the arguments presented in \citet{darden08} (Sections 3.5.2.3.2 and 3.5.2.3.1), one can define
\begin{equation}
  \label{eq:zeta}
  \begin{split}
    \zeta_{k}(\rv) 
    =& \sum_{\nv\in \Omega(P,k)} \frac{\erfc(\param \abs{\rv + \nv})}{\abs{\rv + \nv}}
    + \frac{1}{\pi V} \sum_{\mv \neq 0} \frac{\exp(-\pi^2 |\mv|^2 / \alpha)}{|\mv|^2}\exp(-2i\pi \mv \cdot \rv)
    	\\
  	&\qquad 
	- \frac{\pi}{\alpha V} + H_{k}(\rv), 
  \end{split}
\end{equation}
such that 
\[
	\sum_{\nv\in \Omega(P,k)} \frac{1}{\abs{\rv + \nv}}
	= 
	\zeta_{k}(\rv)  + o(1),
\]
as $k\to \infty$ and which consists of potential at \( \rv \) that is generated by unit point charges located at the vertices of the lattice indexed by $\nv$ such that \(\nv\in  \Omega(P,k) \).
For sake of completeness, we outline this step in Appendix A where we also give the definition of $H_{k}(\rv)$ in Eqn. \eqref{eq:DefHk}.
Based on \( \zeta_{k} \), we now introduce the function
\begin{equation}
	\label{eq:defphi}
 	\Phi_{k} (\rv) \coloneqq \sum_{1\leq j\leq N} \Lj \zeta_{k}(\rv - \rv_{j}),
\end{equation}
defined everywhere but at the location of the point multipoles located in \( U \).
The function \( \rv \mapsto \Phi_{k} (\rv) \) represents the potential at
\( \rv \in U \) generated by all the multipoles and their images contained in periodic lattice cells indexed by $\nv$ such that
\(\nv\in  \Omega(P,k) \).

In consequence, the limit as \( \rv \) tends to any point multipole \( \rv_j \) is not
finite.
Note that this has been handled above with the $^\prime$ sign after the sum since only the interaction energy is considered.
Instead, if one considers the potential at position \( \rv \) generated by all multipoles
except the multipole located at \( \rv_i \), then one has to substract the contribution for
\( \nv=0 \) in equation~\eqref{eq:zeta} for \( \zeta_{k}(\rv - \rv_{i}) \) to get
\begin{equation}
  \label{eq:potentialrj}
  \phi_{k}^{i}(\rv) = \left( \Phi_{k}(\rv) - \Li \frac{1}{\abs{\rv-\rv_{i}}} \right),
\end{equation}
with finite limit at \( \rv_{i} \) given by
\begin{equation}
  \phi_{k}^{i}(\rv_{i})
  = \lim_{\rv \to \rv_{i}}\phi_{k}^{i}(\rv)
  = \lim_{\rv \to \rv_{i}} \left( \Phi_{k}(\rv) - \Li\frac{1}{\abs{\rv-\rv_{i}}} \right).
\end{equation}
The function \( \rv\mapsto \phi_{k}^{i}(\rv) \) denotes the potential at an arbitrary
position \( \rv \) generated by all multipoles 
contained in periodic lattice cells indexed by $\nv$ such that
\(\nv\in  \Omega(P,k) \)
except multipole \( i \) in the unit cell ($\nv=0$).

Hence, using the splitting introduced in equation~\eqref{eq:zeta} combined with~\eqref{eq:defphi} and~\eqref{eq:potentialrj}, it follows that

\begin{multline}
  \label{eq:potentielr}
  \phi_{k}^{i}(\rv)
  = \sideset{}{^{\prime}}\sum_{\substack{\nv\in \Omega(P,k)\\1\leq j \leq N}} \Lj \frac{\erfc(\param \abs{\rv_{j\nv}})}{\abs{\rv_{j\nv}}}
  + \frac{1}{\pi V} \sum_{\mv \neq 0} \frac{\exp(-\pi^2 |\mv|^2 / \alpha)}{|\mv|^2}\exp(-2\pi \ii \mv \cdot \rv)\oS(\mv) \\
  - \left( \Li \frac{\erf(\param \abs{\rv-\rv_{i}})}{\abs{\rv-\rv_{i}}} \right)
  + \sum_{1 \leq j \leq N}\Lj H_{k}(\rv-\rv_{j}),
\end{multline}
with \( r_{j\nv} \coloneqq \abs{\rv-\rv_{j}+\nv} \) and where \( \oS(\mv) \) was defined in
equation~\eqref{eq:struct_fourier}.

We have therefore introduced the splitting (precisely defined in the listing below)
\begin{equation}
  \label{eq:potentiel_f}
  \phi_{k}^{i}(\rv) = \phi_{0,k}^{i}(\rv) + \phi_{\mself}^{i}(\rv)+ \phi_{k,{\msurf}}^{i}(\rv),
\end{equation}
where each individual term is defined and discussed in the following. 

\begin{description}
\item [The absolutely converging part of the potential] Due to their quick convergence in $k\to\infty$ the
first two terms in equation~\eqref{eq:potentielr}, denoted by \( \phi_{0,k}^{i} \)
in equation~\eqref{eq:potentiel_f}, do not depend on the order of summation.
The first term is called the direct potential and the second the reciprocal potential.

\item [The self-potential] The third term in equation~\eqref{eq:potentielr} is what we call
the self potential \( \phi_{\mself}^{i}(\rv) \) in~\eqref{eq:potentiel_f} and
does not depend on the other nuclear position \( \rv_{i} \), \( i\neq j \), and $k$, and is non-constant in $\rv$ around $\rv_i$.
This term being independent on the other nuclear positions can by no means model the
interaction potential, hence the name.

From the derivation it becomes clear that in the limit \( \rv\to \rv_{i} \), \( \phi_{\mself}^{i}(\rv_i) \) is the quantity to
be subtracted from the reciprocal potential in order that the potential at \( \rv=\rv_{i} \)
is the potential created by all other multipoles except multipole \( i \).
Note that the contribution in the direct space has already been taken into account in
equation~\eqref{eq:potentielr} since the sign \( {}^{\prime} \) appears on the first sum.

We will provide in Section \ref{sec:self} explicit values of this terms in limit $\rv\to \rv_i$ for arbitrary multipolar distributions. 

\item [The surface-potential] The fourth term~\eqref{eq:potentielr}, denoted by
\( \phi_{k,\msurf}^{i} \) in equation~\eqref{eq:potentiel_f}, is the
surface potential which will be well-defined only if the sum converges as \( k \) tends to
infinity.
It is intimately linked with the order of summation and is related to subtle questions. We want to focus in this article to  the self-terms and are therefore assuming convergence in $k$ here.
\end{description}

\begin{rem}\label{rem:ms}

It is not very intuitive to understand what is meant by the macroscopic shape of the system
and its environment, and how this is mathematically accounted for.
From the microscopic viewpoint, the sequence of shapes~\( {(\Omega(P,k))}_{k\in\N} \) should be
seen as the scaling of one macroscopic shape~\( P \), i.e.~\( \Omega(P,k) = k P \).
Then, the sequence \( \Omega(P,k) \) covers larger and larger parts of the microscopic
space~\( \mathbb R^3 \) as \( k \) increases.
We would like to advocate also the viewpoint of introducing a change of variables from the
microscopic variable~\( \nv \) to the macroscopic variable~\( \hat \nv = \nv/k \) that can be used
to rewrite sums of the form
\begin{equation}
  \sum_{\nv\in \Omega(P,k) \cap \mathbb Z^3} f(\nv)
  \quad\text{as}\quad
  \sum_{\hat \nv\in P\cap ( \mathbb Z^3/k) } \hat f( \hat \nv),
  \quad \text{with \( \hat f( \hat \nv) \coloneqq f( k\hat \nv)\).}
\end{equation}
This means that the microscopic space contracts more and more within the macroscopic
shape~\( P \), see Figure~\ref{fig:macro} for an illustration.
The role of the macroscopic shape~\( P \) becomes  visible and the exterior of~\( P \) is
then the surrounding environment to~\( P \).
\end{rem}

\begin{figure}[htbp!]
  \centering
  \subcaptionbox{Microscopic space}{%
    \includegraphics[width=0.48\textwidth]{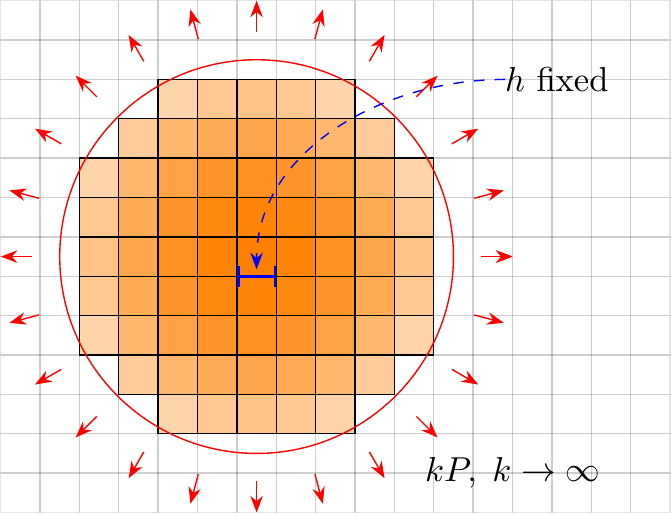}
  }
  \subcaptionbox{Macroscopic space}{%
    \includegraphics[width=0.48\textwidth]{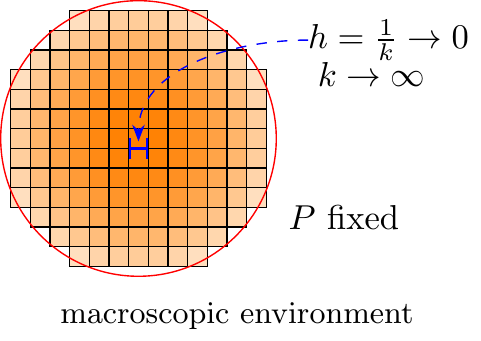}
  }
  \caption{The limiting process \( k\to\infty \) observed in the microscopic and macroscopic
    space}\label{fig:macro}
\end{figure}

In the following, we introduce
\begin{equation}
  \phi_0^i(\rv) = \lim_{k\to\infty} \phi_{0,k}^i(\rv), \qquad
  \phi_{{\msurf}}^i(\rv) = \lim_{k\to\infty} \phi_{k,{\msurf}}^i(\rv),
\end{equation}
where we have assumed that the second term converges as we want to study the self-terms.

Thus
\begin{equation}
	\label{eq:PotSplit}
  \phi^i(\rv)= \phi_{0}^i(\rv) + \phi_{\mself}^i(\rv) + \phi_{{\msurf}}^i(\rv),
\end{equation}
denotes the potential at position \( \rv \) generated by all multipoles except multipole
\( i \) in the unit cell.
Recall that the term \( \phi_{\msurf}^i(\rv) \) depends on the order of
summation represented by a particular shape $P$, whereas the other terms \( \phi_{0,k}^i(\rv) \) and 
\( \phi_{{\mself}}^i(\rv) \) do not.

\section{Derivation of the field}\label{sec:derivation-field}

The derivation we have given for the potential gives a straightforward one for the field.
It is based on the splitting of \( \phi_{k}^{i}(\rv) \) developed in the previous section
and uses the fact that the electric field is minus the gradient of the electric potential.
Indeed, taking the derivative \( \Dif_{\rv} \) (thus with respect to \( \rv \)) in
equation~\eqref{eq:potentielr} yields
\begin{equation}
  \label{eq:div-chp}
  \begin{split}
    \vc \fld_{k}^{i}(\rv) &= - \Dif_{\rv} \phi_{k}^{i}(\rv) \\
    &= -\sideset{}{^{\prime}}\sum_{\mathclap{}\substack{\nv \in \Omega(P,k) \\ 1\leq j \leq N}} \Dif_{\rv}\Lj \frac{\erfc(\param\abs{\rv_{j\nv}})}{\abs{\rv_{j\nv}}}
    - \frac{1}{\pi V} \sum_{\mathclap{\mv \neq 0}} \frac{\exp(-\pi^2 \mv^2 / \alpha)}{\mv^2}
    \oS(\mv)\Dif_{\rv}\!\exp(-2\pi  \ii \mv \cdot \rv) \\
    &\qquad
    + \Dif_{\rv}\Li\frac{\erf(\param \abs{\rv-\rv_{i}})}{\abs{\rv-\rv_{i}}} 
    - \sum_{\mathclap{\substack{1 \leq j \leq N}}} \Dif_{\rv} \Lj H_{k}(\rv-\rv_{j}).
  \end{split}
\end{equation}
Therefore, \( \vc \fld_{k}^{i}(\rv) \) denotes the electric field at a general position \( \rv \)
generated by all multipoles in a cell belonging to \( \Omega(P,k) \) except multipole \( i \)
in the unit cell.

In consequence, we define each term individually as for the potential:
\begin{subequations}
  \begin{align}
    \begin{split}
      \fld_{0,k}^i(\rv)
     & = -\sideset{}{^{\prime}}\sum_{\mathclap{}\substack{\nv \in \Omega(P,k) \\1\leq j \leq N}} \Dif_{\rv}\Lj \frac{\erfc(\param\abs{\rv_{j\nv}})}{\abs{\rv_{j\nv}}}\\
      &\qquad\qquad- \frac{1}{\pi V} \sum_{\mathclap{\mv \neq 0}} \frac{\exp(-\pi^2 \mv^2 / \alpha)}{\mv^2}
     \oS(\mv) \Dif_{\rv}\!\exp(-2\pi \ii \mv \cdot \rv),
    \end{split}\\
    \fld_{\mself}^i(\rv)
    &= \Dif_{\rv}\Li \frac{\erf(\param \abs{\rv-\rv_{i}})}{\abs{\rv-\rv_{i}}},\\
    \fld_{{\msurf},k}^i(\rv)
    &= - \sum_{\mathclap{\substack{1 \leq j \leq N}}} \Dif_{\rv} \Lj H_{k}(\rv-\rv_{j}).
  \end{align}
\end{subequations}

In particular, we defined the self electric field \( \fld_{\mself}^i(\rv) \) as the third term in~\eqref{eq:div-chp}, which will be shown in  Section~\ref{sec:self} to be well-defined, in particular at \( \rv_i \), and give explicit expressions.
Evaluating \( \vc \fld_{k}^{i}(\rv) \) at \( \rv=\rv_{i} \) then yields
\begin{equation}
  \label{eq:field_f}
  \fld_{k}^{i}(\rv_{i}) =
  \fld_{0,k}^{i}(\rv_{i}) + \fld_{\mself}^{i}(\rv_{i}) + \fld_{k,{\msurf}}^{i}(\rv_{i}).
\end{equation}

Using classical results from convergence of series, we obtain that as soon as the
surface-potential converges (in the limit \( k\to\infty \)) and the surface-field
converges uniformly in \( \rv \) in a neighborhood of \( \rv_i \), 
the gradient of the limit of the surface-potential is exactly the surface-field:
\begin{equation}
  \fld_{\msurf}^{i}(\rv)
 = - \lim_{k\to\infty}\sum_{\mathclap{\substack{1 \leq j \leq N}}} \Dif_{\rv} \Lj H_{k}(\rv-\rv_{j})
  = - \Dif_{\rv} \left( \lim_{k\to\infty}\sum_{\mathclap{\substack{1 \leq j \leq N}}} \Lj H_{k}(\rv-\rv_{j}) \right)
  = - \Dif_{\rv} \left( \phi_{\msurf}^{i}(\rv) \right).
\end{equation}
Again, this is a subtle question related to the convergence in $k$ that will be addressed in an upcoming work.
The focus of this article is shed on the self-terms.

\section{Derivation of the energy}\label{sec:derivation-energy}

Recalling equation~\eqref{eq:Equiv} combined with the splitting~\eqref{eq:PotSplit} of the potential into different parts,
we define the following energy contributions
\begin{subequations}
  \begin{align}
    \mathcal{E}_{0}(\urvN)
    &= \frac12 \sum_{1\leq i \leq N} \sum_{0\leq k \leq p} \vs M^{k}_{i} \cdot \big(  \Dif^{k}_{\rv } \phi_0^i(\rv)\big)\big|_{\rv=\rv_i},\\
    \label{eq:SelfEn}
    \mathcal{E}_{\mself}(\urvN)
    &= \frac12\sum_{1\leq i\leq N} \sum_{0\leq k \leq p} \vs M^{k}_{i} \cdot \big(  \Dif^{k}_{\rv } \phi_\mself^i(\rv)\big)\big|_{\rv=\rv_i},\\
    \mathcal{E}_{\msurf}(\urvN)
    &= \frac12 \sum_{1\leq i \leq N} \sum_{0\leq k \leq p} \vs M^{k}_{i} \cdot \big(  \Dif^{k}_{\rv } \phi_\msurf^i(\rv)\big)\big|_{\rv=\rv_i}.
  \end{align}
\end{subequations}
Note that the self-potential is non-constant in a neighborhood of $\rv_i$ so that the higher multipolar moments, i.e. the derivatives, act on the self-potential $\phi_\mself^i$. 
Further, notice that $\mathcal{E}_{0}(\urvN)$ can be written as
\begin{multline}
	\mathcal{E}_{0}(\urvN) = \frac12 \sum_{1\leq i \leq N}  \Li \phi_0^i(\rv_i)
	= 
	\frac{1}{2}
	\sideset{}{^{\prime}}\sum_{\mathclap{}\substack{\nv \\ 1\leq i,j \leq N}} \Li \Lj \frac{\erfc(\param\abs{\rijn})}{\abs{\rijn}}
	\\
	+
	\frac{1}{2\pi V} \sum_{\mv \neq \vc 0} \frac{\exp(-\pi^2\mv^2/\alpha)}{\mv^2}\oS(\mv)\oS(-\mv),
\end{multline}
and in consequence, we write
\begin{equation}
  \mathcal{E}(\urvN)
  = \underbrace{\mathcal{E}_{0}(\urvN) + \mathcal{E}_{\mself}(\urvN)}_{\mathcal{E}_{\mEwald}(\urvN)}
  + \underbrace{\mathcal{E}_{\msurf}(\urvN)}_{J(\vs D,\vs M)}
  = \mathcal{E}_{\mEwald}(\urvN) +J(\vs D,\vs M).
\end{equation}
Note that we confirm with this derivation equation~\eqref{eq:energy} and that the \Ewald\ energy and the self-energy do not
depend on the order of summation whereas the surface energy does.

The corresponding force-terms then naturally result from differentiating the different
energies with respect to the nuclear coordinates.
In particular, as we will see further below, the self-energy is independent on any nuclear
coordinate and the self-energy therefore doesn't induce any force term.
However, the correct term of the self-field is mandatory in the context of polarizable
force-fields.\cite{Lagardere2015ScalableEwald}

\section{Well-posedness of the self-terms}\label{sec:self}

In this part, we outline the proofs that the self-potential, field and \eqref{eq:SelfEn} are well-defined in the limit $\rv\to\rv_j$ and in consequence also the self-energy.
As done by \citet{smith98}, we introduce recursively the
functions \( B_n \) for any \( n \in \N \) and all \( r \in \R_{+}\setminus\{0\} \) by
\begin{equation}
	\label{eq:defBn}
  \begin{split}
    B_{0}(r) &\coloneqq - \frac{\erf(\param r)}{r} \\
    B_n(r) &\coloneqq \frac{1}{r^2} \left( (2n-1) B_{n-1}(r) + \frac{{(2\alpha)}^n}{\sqrt{\alpha\pi}} \exp(-\alpha r^{2}) \right).
  \end{split}
\end{equation}

Then, the following result holds.
\begin{thm}
\label{thm}
For any \(  n \in \N \), there holds that
\begin{equation}
  \label{eq:self}
  \lim_{r \to 0} B_n(r) = B_n(0) = -\frac{\alpha^{n+1/2}}{\sqrt\pi} \frac{2^{n+1}}{2n+1},
\end{equation}
and
\begin{equation}
	\label{eq:self2}
  \frac{\dd B_{n}}{\dd r}(r) = -r B_{n+1}(r).
\end{equation}
\end{thm}
The proof of Theorem~\ref{thm} is presented in Appendix B.



In order to give explicit formulae for the self-terms, we first note that from \eqref{eq:self2} follows
\begin{equation}
  \Dif_i^1 B_{n}(\abs{\rv - \rv_{i}}) = (\rv-\rv_i) B_{n+1}(\abs{\rv - \rv_{i}}).
\end{equation}
Since we have derived the values of \( B_n(0) \) in \eqref{eq:self}, we can give explicit
formulae for the self-potential, the self-field and the self-energy in consequence.

For sake of a simple presentation, we consider a multipolar charge distribution up to
quadrupoles in the following.
Intrinsically, the quadrupolar moments \(\vs M_i^2 \) are symmetric matrices with zero trace.

\begin{description}
\item [The self-potential] 
Therefore, the \( i \)-th self-potential \( \phi_{\mself} \) at an arbitrary
point \( \rv \) in a neighborhood of \( \rv_i \) writes as
\begin{subequations}
  \begin{align}
    \phi_{\mself}^i(\rv)
    &= -\Li \frac{\erf(\param \abs{\rv-\rv_{i}})}{\abs{\rv-\rv_{i}}}
      = \Li B_0(\abs{\rv-\rv_{i}}) \\
    &=  	\vs M^0_i  \, B_0(\abs{\rv-\rv_{i}})
      	+	\vs M^1_i \cdot \Dif^1_i  \, B_0(\abs{\rv-\rv_{i}})
      	+ 	\vs M^2_i \cdot \Dif^2_i  \, B_0(\abs{\rv-\rv_{i}}) \\
    \begin{split}
      &= 
      		\vs M^0_i  \, B_0(\abs{\rv-\rv_{i}}) 
      	+	B_1(\abs{\rv-\rv_{i}})   \, \vs M^1_i \cdot (\rv-\rv_{i}) \\
      &\qquad\qquad + B_2(\abs{\rv-\rv_{i}})  \, \vs M^2_i \cdot \big({(\rv-\rv_{i})}^\top (\rv-\rv_{i}) \big).
    \end{split}
  \end{align}
\end{subequations}

Then, the evaluation of the \(i \)-th self-potential at \( \rv=\rv_i \) is  given by
\begin{equation}
  \phi^i_{\mself}(\rv_{i})
  = \lim_{\rv \to  \rv_{i} } \phi_{\mself}^i(\rv)
  = \vs M^0_i \, B_0(0)
  = -\vs M^0_i  \, 2 \sqrt{\frac{\alpha}{\pi}} ,
\end{equation}
which does no longer depend on \( \rv_{i} \), only depends on the charge \(\vs M_i^0 \) and
is well-defined.
This formula is of course valid for any kind of multipolar distribution and not restricted
to orders to up to quadrupoles only.
Note that the self-potential is not constant in a neighborhood of $\rv_i$ in this derivation.

\item[The self-field]
We want to stress that in contrast to what is presented in~\citet{nymand00}, there is  indeed a non-zero self-contribution to the electric field. 
The $i$-th part of the self-field at \( \rv=\rv_i \) is defined by
\begin{equation}
  \fld_{\mself}^i(\rv)
  = - \lim_{\rv\to \rv_i} \left(\Dif_{\rv} \phi^i_{\mself}(\rv) \right)
  = - \vs M^1_i   \, B_1(0)
  = \vs M^1_i \, \sqrt{\frac{\alpha}{\pi}}\frac{4\alpha}{3},
\end{equation}
which only depends on the dipole moment at site \( i \) and is also valid for any kind of
multipolar distribution.

\item[The self-energy]
Finally, the self energy as defined above writes as
\begin{align}
  \mathcal{E}_{\mself}(\urvN)
  & = \frac12\sum_{1\leq i\leq N} \sum_{0\leq k \leq 2} \vs M^{k}_{j} \cdot \big(  \Dif^{k}_{\rv } \phi_\mself^i(\rv)\big)\big|_{\rv=\rv_i},
  \\
  & = \frac12\sum_{1\leq i\leq N}\left(
    \vs M^0_i\cdot \vs M^0_i  \, B_0(0) + \vs M^1_i\cdot \vs M^1_i  \, B_1(0) + 2  \, \vs M^2_i \cdot \vs M^2_i \, B_2(0)
    \right) \\
  & = - \sqrt{\frac{\alpha}{\pi}} \sum_{1\leq i\leq N}\left(
    \vs M^0_i\cdot \vs M^0_i + \frac{2\alpha}{3} \vs M^1_i\cdot \vs M^1_i + \frac{8\alpha^2}{5}   \vs M^2_i \cdot \vs M^2_i
  \right).
\end{align}
\black 

We recognize the current practice that for relative energies and forces, the correct term of
the self-energy is not needed since a constant misfit cancels out in energy differences.
However, for sake of having a complete theory based on a rigorous development, we think that
it is important to state the self-energy as well.
\end{description}


\section{Conclusion}
In that paper, we proposed a new mathematically clean and coherent derivation of the \Ewald\ summation
for a system consisting of \( N \)\tiret{}body electrostatic interaction with multipolar
charges of any order.
The existing results in the literature differ between different authors and no common
development of all quantities can be found.
The essential differences lie in the self-term expressions.
We presented a clean derivation and confirm the expressions proposed by \citet{smith98} for which we proved
well-posedness.
Our model is derived from a clean application of the \Ewald\ splitting to the electric
potential and the subsequent quantities such as the electric field, the energy and the
forces are derived thereof.
A complete derivation of all these quantities is mandatory in the context of next generation
polararizable force-fields where in particular the self-field is required and needs to be
consistent with the theory.

Overall, the new model which is mathematically sound maintains the use of the tinfoil model
and provides simpler expressions for the self-energy that are closer to the original idea of
\Ewald\ to work on the potential and not on the energy.

\section*{acknowledgements}
This work was made possible thanks to the French state funds managed by
the CalSimLab LABEX and the ANR within the Investissements d’Avenir
program (reference ANR-11-IDEX-0004-02) and through support of the
Direction Générale de l’Armement (DGA) Maîtrise NRBC of the French
Ministry of Defense.

Benjamin Stamm acknowledges the funding from the German Academic Exchange Service (DAAD)
from funds of the ``Bundesministeriums f$\ddot{\mbox u}$r Bildung und Forschung'' (BMBF) for the project Aa-Par-T (Project-ID 57317909).

Yvon Maday and Etienne Polack acknowledge the funding from the PICS-CNRS (Project ${\text N}^{\circ}$ 230509) and the PHC PROCOPE 2017 (Project ${\text N}^{\circ}$ 37855ZK).

\bibliography{self}

\section*{Appendix A}
\setcounter{equation}{0}
\renewcommand\theequation{A\arabic{equation}}
Essentially, for sake of a complete presentation we present here the derivation of \eqref{eq:zeta} in a compact way following the arguments presented in \citet{darden08}, see also \cite{Smith1981,Smith1994}.
As briefly mentioned in Section~\ref{sec:derivation-pot}, we start with the following splitting
\begin{equation}
  \label{eq:distanceSplit2}
  \frac{1}{\abs{\rv}}
  = 
  \frac{\erfc(\param \abs{\rv})}{\abs{\rv}}
  + \frac{1}{\pi} \sum_{\mv} \int_{U^*} \frac{\exp(-\pi^2 \abs{\vv + \mv }^2 / \alpha)}{\abs{\vv + \mv}^2}
  \exp(-2 \pi \ii (\vv + \mv) \cdot \rv) \, \dd[3] \vv,
\end{equation}
see Eq. (3.5.2.16) in \citet{darden08}.
Then, 
one can write
\begin{align*}
    \sum_{\nv\in \Omega(P,k)} \frac{1}{\abs{\rv + \nv}}	
    &=
    \sum_{\nv\in \Omega(P,k)} \frac{\erfc(\param \abs{\rv+\nv})}{\abs{\rv+\nv}}
  + 
   	\sum_{\nv\in \Omega(P,k)}
  	\sum_{\mv} 
  	\int_{U^*} h_{\mv,\rv}(\vv) \exp(-2 \pi \ii \vv \cdot \nv) \, \dd[3] \vv
\end{align*}
with
\[
	h_{\mv,\rv}(\vv) 
	= 
	\frac{\exp(-\pi^2 \abs{\vv + \mv }^2 / \alpha)}{\pi\abs{\vv + \mv}^2}
  \exp(-2 \pi \ii (\vv + \mv) \cdot \rv),
\]
and where we used that
\[
	\exp(-2 \pi \ii (\vv + \mv) \cdot \nv) 
	=
	\exp(-2 \pi \ii \vv  \cdot \nv) 
\]
since $ \mv \cdot \nv\in \mathbb N$. 

\noindent
{\bf Case $\mv\neq 0$:}
We first recognize that 
\[
	\frac{1}{V}\widehat h_{\mv,\rv}(\nv)
	=
	\int_{U^*} 
  	h_{\mv,\rv}(\vv)
  	\exp(-2 \pi \ii\vv  \cdot \nv) \, \dd[3] \vv
\]
where $\widehat h_{\mv,\rv}$ denotes the Fourier coefficient of $ h_{\mv,\rv}$ and we recall that $V= \frac{1}{|U^*|}$.
Then, there holds that 
\[
	\sum_{\nv} \widehat h_{\mv,\rv}(\nv) =  h_{\mv,\rv}(0)
	= \frac{\exp(-\pi^2 \abs{\mv }^2 / \alpha)}{\pi\abs{\mv}^2}
  \exp(-2 \pi \ii\mv \cdot \rv)
\]
and thus 
\[
	\sum_{\nv\in \Omega(P,k)} \widehat h_{\mv,\rv}(\nv)
	= h_{\mv,\rv}(0) + o(1),
\]
as $k\to\infty$.

\noindent
{\bf Case $\mv = 0$:}
As visible from above, this development does not hold for $\mv=0$ and is more subtle.
We have that
\[
	h_{\mv,\rv}(\vv) 
	= 
	\frac{\exp(-\pi^2 \abs{\vv}^2 / \alpha)}{\pi\abs{\vv}^2}
  \exp(-2 \pi \ii \vv  \cdot \rv)
\]
and the combination of two Taylor expansions yields
\begin{align*}
	h_{\mv,\rv}(\vv) 
	&= 
	\frac{1}{\pi\abs{\vv}^2}
	\left( 1- \frac{\pi^2}{\alpha} \abs{\vv}^2 +\mathcal O(|\vv|^4)\right)
	\left( 1-2\pi \ii \vv\cdot \rv - 2\pi^2\abs{\vv\cdot\rv}^2 + \mathcal O(\abs{\vv}^3) \right)
  \\
	&= 
	\frac{ 1-2\pi \ii\vv\cdot \rv - 2\pi^2\abs{\vv\cdot\rv}^2 }{\pi\abs{\vv}^2}
	 - 
	\frac{\pi}{\alpha} 
  + 
  \mathcal O(|\vv|).
\end{align*}
This motivates the definition
\begin{equation}
	\label{eq:DefHk}
	H_k(\rv) 
	= 
	\sum_{\nv\in \Omega(P,k)} 
	\int_{U^*} 	
	\frac{1-2\pi \ii \vv\cdot \rv - 2\pi^2\abs{\vv\cdot\rv}^2  }{\abs{\vv}^2} 
	\exp(-2 \pi \ii \vv  \cdot \nv )
	 \, \dd[3] \vv,
\end{equation}
and note that 
\[
	- 
	\frac{\pi}{\alpha} 
	\sum_{\nv}
  	\int_{U^*} 
  	1
  	\exp(-2\pi \ii \vv  \cdot \nv) \, \dd[3] \vv
  	=
	- 
	\frac{\pi}{\alpha V} 
	\sum_{\nv}
  	\widehat 1(\nv)
  	=
  	- 
	\frac{\pi}{\alpha V} .
\]
Then 
\[
	- 
	\frac{\pi}{\alpha} 
	\sum_{\nv\in \Omega(P,k)}
  	\int_{U^*} 
  	\exp(-2 \pi \ii \vv  \cdot \nv) \, \dd[3] \vv
  	=
  	- 
	\frac{\pi}{\alpha V}  
	+ o(1),
\]
as $k\to\infty$.
Further there holds
\[
	\sum_{\nv\in \Omega(P,k)} \int_{U^*} \mathcal O(|\vv|) \exp(-2 \pi \ii \vv  \cdot \nv)\, \dd[3] \vv
	=
	 o(1),
\]
so that combining all terms yields
\begin{align*}
	\zeta_{k}(\rv) 
    &=
    \sum_{\nv\in \Omega(P,k)} \frac{\erfc(\param \abs{\rv+\nv})}{\abs{\rv+\nv}}
   	+
   	\frac{1}{\pi V}
   	\sum_{\mv\neq 0} 
  	\frac{\exp(-\pi^2 \abs{\mv }^2 / \alpha)}{\abs{\mv}^2}
  	\exp(-2\pi\ii\mv \cdot \rv)
  	\\
  	&\qquad 
  	- 
	\frac{\pi}{\alpha V} 
	+
	\frac{1}{V} H_k(\rv) + o(k),
\end{align*}
as $k\to\infty$.
Note that we do not shed emphasis on the different arguments that guarantee existence of the different limits but put rather emphasis on the compact development to derive \eqref{eq:zeta}.

\section*{Appendix B: Mathematical proofs}
\setcounter{equation}{0}
\renewcommand\theequation{B\arabic{equation}}
\label{ssec:proofs}
Before we really tackle the proof of Theorem~\ref{thm}, we first prove some auxiliary results.

\begin{lemma}
For any integer \( n \), the function \( B_n \) is explicitly
given by
\begin{equation}
  \label{eq:Bn1}
  B_n(r)
  = \frac{\exp(-\alpha r^2)}{\sqrt{\alpha\pi} r^2}\sum_{k=0}^{n-1}
  \frac{{(2\alpha)}^{n-k}}{r^{2k}} \frac{(2n-1)!!}{(2(n-k)-1)!!}
  - (2n-1)!! \frac{\erf(\param r)}{r^{2n+1}},
\end{equation}
where \( (2n-1)!! \coloneqq (2n-1)\times\cdots\times 3\times 1 \) 
with the convention that
for any non-positive integer \( k \), \( k!! = 1 \) and that a sum from \( 0 \) to \( -1 \)
is zero.
\end{lemma}
\begin{proof}
The proof follows by induction. For $n=0$ we see that the proposition holds by inspection.
Now, let us assume that~\eqref{eq:Bn1} holds for a given $n$. 
Inserting~\eqref{eq:Bn1}  into the definition of $B_{n+1}$ in (\ref{eq:defBn}) implies
\begin{align*}
r^2 & B_{n+1}(r) =  (2n+1) B_{n}(r) + \frac{{(2\alpha)}^{n+1}}{\sqrt{\alpha\pi}} \exp(-\alpha r^{2}) \\
&= (2n+1) \left(\frac{\exp(-\alpha r^2)}{\sqrt{\alpha\pi} r^2}\sum_{k=0}^{n-1}
  \frac{{(2\alpha)}^{n-k}}{r^{2k}} \frac{(2n-1)!!}{(2(n-k)-1)!!}
  - (2n-1)!! \frac{\erf(\param r)}{r^{2n+1}}\right) \\
  	&\qquad \qquad\qquad + \frac{{(2\alpha)}^{n+1}}{\sqrt{\alpha\pi}} \exp(-\alpha r^{2}) \\
&=  \frac{\exp(-\alpha r^2)}{\sqrt{\alpha\pi}} \left( {(2\alpha)}^{n+1} + \sum_{k=1}^{n}
  \frac{{(2\alpha)}^{n+1-k}}{r^{2k}} \frac{(2n+1)!!}{(2(n+1-k)-1)!!}\right)
  - (2n+1)!! \frac{\erf(\param r)}{r^{2n+1}} \\
&=  \frac{\exp(-\alpha r^2)}{\sqrt{\alpha\pi}}  \sum_{k=0}^{n}
  \frac{{(2\alpha)}^{n+1-k}}{r^{2k}} \frac{(2n+1)!!}{(2(n+1-k)-1)!!}
  - (2n+1)!! \frac{\erf(\param r)}{r^{2n+1}} 
\end{align*}
\end{proof}

\begin{lemma}
The functions \( B_n \) can be rewritten for all positive \( r \) as
\begin{multline}
  \label{eq:Bn2}
 B_n(r)
  = \frac{1}{\sqrt{\alpha\pi}}
  \sum_{\ell=0}^\infty
  \sum_{k=0}^{n-1}
  \frac{{(-1)^\ell \alpha}^{\ell+n-k} 2^{n-k}(2n-1)!!}{\ell!(2(n-k)-1)!!} 
  r^{2(\ell-k-1)}
  \\
   \quad 
  - 2 (2n-1)!! \sqrt{\frac{\alpha}{\pi}} \sum_{\ell=0}^\infty\frac{(-\alpha)^\ell }{(2\ell+1)\ell!} r^{2(\ell -n )}.
\end{multline}
\end{lemma}

\begin{proof}
The result is obtained by inserting the expression of the error and exponential functions as a power series, i.e.,
\[
	\exp(-\alpha r^2)
	=
	\sum_{\ell=0}^\infty \frac{(-\alpha)^\ell r^{2\ell}}{\ell!}
	\qquad\mbox{and}\qquad
	\erf(\param r)
	=
	2\sqrt{\frac{\alpha}{\pi}} \sum_{\ell=0}^\infty\frac{(-\alpha)^\ell r^{2\ell+1}}{(2\ell+1)\ell!},
\]
in equation~\eqref{eq:Bn1}. 
\end{proof}

We now need to prepare some result that is used in a later proof.
\begin{lemma}
\label{lem:basicId}
For any $n\in \mathbb N_0$, there holds
\[
 \sum_{k=1}^{n+1}
  \frac{(-2)^{k} }{(n-k+1)!(2k-1)!!} 
+ \frac{2}{(2n+1)n!} 
=
0.
\]
\end{lemma}
\begin{proof}
We denote by $\Gamma$ the usual gamma-function. Introduce as well the double-factorial for even numbers $(2n)!! =2n \times  (2n-2) \times \ldots \times 2$ and observe that the following identities hold
\begin{align*}
(2k)!! &= 2^k \, k! = 2^k \, \Gamma(k+1),\\
(2k+1)!! &= \frac{(2k+2)!}{(2k+2)!!}
=\frac{(2k+2)!}{2^{k+1}(k+1)!}= 2^{-k-1}\frac{\Gamma(2k+3)}{\Gamma(k+2)}.
\end{align*}
In consequence, there holds
\begin{equation}
	\label{eq:dev1}
	\frac{(2k)!!}{(2k+1)!!}
	=
	2^{2k+1}\frac{\Gamma(k+1)\Gamma(k+2)}{\Gamma(2k+3)} = 2^{2k+1}B(k+1,k+2)
	=
	2^{2k+1} \int_0^1 t^{k+1} (1-t)^k dt,
\end{equation}	
where $B(\cdot,\cdot)$ denotes the beta-function.
Next, we observe that 
\begin{equation}
	\label{eq:dev2}
	2\int_0^1  t \big(t^{k} (1-t)^k \big)dt = \int_0^1 t^{k} (1-t)^k dt,
\end{equation}	
by exploiting the change of variable $s=1-t$ and the fact that $t(1-t)=s(1-s)$.
Further using the identity $4t(1-t) = 1-(2t-1)^2$ and another change of variable $s=2t-1$ yields
\begin{equation}
	\label{eq:dev3}
	4^k
	\int_0^1 t^{k} (1-t)^k dt
	=
	\int_0^1 \big( 1-(2t-1)^2 \big)^k dt
	=
	\frac12
	\int_{-1}^1 \big( 1-s^2 \big)^k ds
	=
	\int_{0}^1 \big( 1-s^2 \big)^k ds.
\end{equation}	
Combining \eqref{eq:dev1}--\eqref{eq:dev3} then yields
\begin{equation}
	\label{eq:dev4}
	\frac{(2k)!!}{(2k+1)!!}
	=
	\int_{0}^1 \big( 1-t^2 \big)^k dt.
\end{equation}
Now, we use \eqref{eq:dev4} in combination with the the binomial coefficient theorem as follows
\begin{equation}
	\label{eq:dev5}
	\sum_{k=0}^{n}
  \frac{(-1)^{k} (2k)!! }{(2k+1)!!} 
  \frac{n!}{(n-k)!k!}
  	= 
  	\int_{0}^1 \sum_{k=0}^{n}
  \frac{n!}{(n-k)!k!} \big( t^2 -1 \big)^k dt
   =
  	\int_{0}^1 t^{2n} dt
  	=
  	\frac{1}{2n+1}.
\end{equation}
Then replacing the double factorial $(2k)!!=2^k k!$ and shifting the index $k$ by one yields the desired result.
\end{proof}

The first claim of Theorem~\ref{thm} is formulated in the following lemma.
\begin{lemma}
For any \(  n \in \N \), there holds that
\begin{equation}
  \label{eq:self3}
  \lim_{r \to 0} B_n(r) = B_n(0) = -\frac{\alpha^{n+1/2}}{\sqrt\pi} \frac{2^{n+1}}{2n+1}.
\end{equation}
\end{lemma}
\begin{proof}
It remains only to study the limit as \( r \) tends to zero in \eqref{eq:Bn2} of the terms where
\( \ell-n \) and \( \ell-k-1 \) are nonpositive, as positive powers of \( r \) will converge to
zero.

In order to have a clear picture, we first reorder the sums over $\ell$ and $k$ of the first term
in \eqref{eq:Bn2} by introducing the following change of indices $\ell=q-s+1$ and $k=n-s$ as follows
\begin{align*}
  \sum_{\ell=0}^\infty
  \sum_{k=0}^{n-1}
  S(\ell,k)
  &=
  \sum_{q=0}^\infty
  \sum_{s=1}^{\min(n,q+1)}
  S(q-s+1,n-s)
  \\
  &=
  \sum_{q=0}^{n-1}
  \sum_{s=1}^{q+1}
  S(q-s+1,n-s)
	+
  \sum_{q=n}^\infty
  \sum_{s=1}^{n}
  S(q-s+1,n-s)
\end{align*}
where $S(\ell,k)$ denotes the summand of the first term in \eqref{eq:Bn2}.
The second term in \eqref{eq:Bn2} is modified by the change of indices $\ell=q-s+1$ and $k=n-s$ resulting in the following expression
\begin{multline}
  B_n(r)
  = 
  \frac{1}{\sqrt{\alpha\pi}}
  \sum_{q=0}^{\infty}
  \sum_{s=1}^{\min(n,q+1)}
  \frac{{(-1)^{q-s+1} \alpha}^{q+1} 2^{s}(2n-1)!!}{(q-s+1)!(2s-1)!!} 
  r^{2(q-n)}
  \\
   \quad 
  - 2 (2n-1)!! \sqrt{\frac{\alpha}{\pi}} \sum_{q=0}^\infty\frac{(-\alpha)^q}{(2q+1)q!} r^{2(q -n )}
  \\
 = 
  -(2n-1)!!\sqrt{\frac{\alpha}{\pi}}
  \sum_{q=0}^{\infty}
  (-\alpha)^q
  \left(
  \sum_{s=1}^{\min(n,q+1)}
  \frac{(-2)^{s} }{(q-s+1)!(2s-1)!!} 
	+\frac{2}{(2q+1)q!} 
	\right) r^{2(q -n )} .
\end{multline}
As outlined in the beginning of the proof, we focus on the non-negative powers of $r$, thus for non-negative $q-n$. 
The coefficient for such a non-negative power $q-n$ of $r$ is given by
\[
 \sum_{s=1}^{q+1}
  \frac{(-2)^{s} }{(q-s+1)!(2s-1)!!} 
+ \frac{2}{(2q+1)q!} 
\]
which vanishes by Lemma~\ref{lem:basicId}.
This proves well-posedness of the limit $r\to 0$ and the limit is given by the coefficient of the zero-th power in $r$, i.e. for $q=n$. Note that $\min(n,q+1)=n=q-1$ and we apply once again Lemma ~\ref{lem:basicId} to obtain the desired limit.
\end{proof}

The second claim of Theorem~\ref{thm} is formulated in the following lemma.
\begin{lemma}
For any \(  n \in \N_0 \), there holds that
\begin{equation}
	\label{eq:drB}
  \frac{\dd B_{n}}{\dd r}(r) = -r B_{n+1}(r).
\end{equation}
\end{lemma}
\begin{proof}
This proof follows by induction. 
For $n=0$, the claim can easily be proven by inspection using the definition \eqref{eq:defBn}.
Consider now the recursive definition of $B_{n+1}$ in \eqref{eq:defBn} and deriving the expression with respect to $r$ yields
\begin{align*}
\frac{\dd B_{n+1}}{\dd r}(r) 
=
-\frac{2}{r} B_{n+1} 
+\frac{1}{r^2}
\left(
(2n+1) \frac{\dd B_{n}}{\dd r}(r) 
-
r
\frac{{(2\alpha)}^{n+2}}{\sqrt{\alpha\pi}} \exp(-\alpha r^{2}) 
\right).
\end{align*}
Assuming that \eqref{eq:drB} holds for $n$ and applying once again the definition \eqref{eq:defBn}  of $B_{n+2}$ implies
\begin{align*}
\frac{\dd B_{n+1}}{\dd r}(r) 
=
-\frac{(2n+1)}{r} B_{n+1} -\frac{1}{r}
\frac{{(2\alpha)}^{n+1}}{\sqrt{\alpha\pi}} \exp(-\alpha r^{2}) 
= - r B_{n+2}(r)
\end{align*}
which completes the proof by induction.
\end{proof}
\end{document}